\documentclass[conference]{IEEEtran}

\usepackage{amsmath,amssymb,amsthm,mathtools,enumitem,xcolor,float}
\usepackage{bm}
\usepackage[margin=1in, bottom=1.05in]{geometry}

\newcommand{\changed}[1]{#1}
\newenvironment{changedblock}{}{}

\newtheorem{theorem}{Theorem}[section]
\newtheorem{lemma}[theorem]{Lemma}

\newtheorem{proposition}[theorem]{Proposition}
\newtheorem{corollary}[theorem]{Corollary}

\theoremstyle{remark}
\newtheorem{remark}{Remark}
\newtheorem{definition}{Definition}[section]
\newtheorem{example}{Example}

\DeclareMathOperator{\dom}{dom}

\begin{document}

\title{Latency-Optimal Geo-Distributed Storage over Structured Networks}

\author{
    \IEEEauthorblockN{Madhura Pathegama and Viveck R. Cadambe}
    \IEEEauthorblockA{School of Electrical and Computer Engineering\\
    Georgia Tech 
    USA\\
    \{macharige3,viveck\}@gatech.edu} 
}

\maketitle
\maketitle

\begingroup
\renewcommand{\thefootnote}{}
\footnotetext{\footnotesize
\copyright~2026 IEEE. Personal use of this material is permitted.
Permission from IEEE must be obtained for all other uses, in any current
or future media, including reprinting/republishing this material for
advertising or promotional purposes, creating new collective works, for
resale or redistribution to servers or lists, or reuse of any copyrighted
component of this work in other works.
}
\addtocounter{footnote}{-1}
\endgroup

 \setlength{\abovedisplayskip}{1pt}
 \setlength{\belowdisplayskip}{1pt}

\begin{abstract}
We study latency-optimal file assignment in geo-distributed storage systems modeled as weighted graphs, where edge weights represent communication delays and each node stores one (possibly coded) file. Our goal is to minimize the average time required to retrieve an original file, taken uniformly over all nodes and files. We show that for every fixed number of files $k \geq 3$, computing a latency-minimizing assignment is NP-hard via a reduction from the domatic number problem. On the positive side, we identify natural network topologies that admit uncoded, structured optimal assignments in which, for every node, one can choose its $k$ closest nodes, including itself, so that they store distinct original files. We prove that every weighted tree, certain weighted cycles, and unit-weight graphs with sufficiently large minimum degree admit such assignments. For these graph classes, we provide efficient algorithms to construct latency-optimal file assignments.
\end{abstract}

% \begin{abstract} 
% We study latency-optimal file assignment in geo-distributed storage systems modeled as weighted graphs, where edge weights represent communication delays and each node stores one (possibly coded) file. Our goal is to minimize the average time required to retrieve an original file, taken uniformly over all nodes and files. We show that for any fixed number of files $k \ge 3$, computing a latency-minimizing assignment is NP-hard via a reduction from the domatic number problem. On the positive side, we identify natural network topologies that admit greedy assignments—a class of uncoded, structured optimal assignments in which each node's $k$ nearest neighbors store distinct original files. We prove that every weighted tree, certain weighted cycles, and unit-weight graphs with sufficiently large minimum degree support greedy assignments. For these graph classes, we provide efficient algorithms to construct latency-optimal file assignments.
% \end{abstract}

\section{Introduction}

Geo-distributed storage systems are widely used in applications such as cloud storage and content delivery networks, where data is stored across geographically dispersed nodes~\cite{baron2026acos,emara2023geographically,uluyol2020near}. 
In such systems, due to the large inter-node latencies, the placement and encoding of data directly affect retrieval latency, since users typically prefer accessing nearby nodes but may need to contact remote locations depending on the storage scheme~\cite{acharya2024existence,uluyol2020near}. 
Designing latency-efficient storage is therefore of interest to both the theory and practice of distributed data storage.

To formalize this problem, the notion of latency-optimal file assignments in geographically distributed storage networks was introduced in~\cite{acharya2024existence}. 
In this setting, $k$ \emph{original files} are stored across $n$ nodes ($n > k$), where each node stores a single \emph{coded file}, which may be either an original file or a linear combination of the original files. 
The network is modeled as a weighted graph, with edge weights representing round-trip times, which we use as our notion of latency.
When a node wishes to retrieve an original file, it contacts other nodes and downloads their stored coded files, from which it reconstructs the desired file.

Our goal is to minimize the average latency of the network, defined as the average time required to retrieve an original file, taken uniformly over all nodes and all original files. The following example illustrates this notion.

\begin{example}\label{ex: sq}
Consider a four-node cycle $\{A,B,C,D\}$ with unit edge lengths and three original files $\{W_1,W_2,W_3\}$. We assign the files
\[
    W_1,\; W_2,\; W_3,\; W_1 + W_2 + W_3
\]
to nodes $A,B,C,D$, respectively, as shown below.

\begin{figure}[H]
    \centering
    \includegraphics[width=1\linewidth, trim={5mm 7mm 5mm 5mm}, clip]{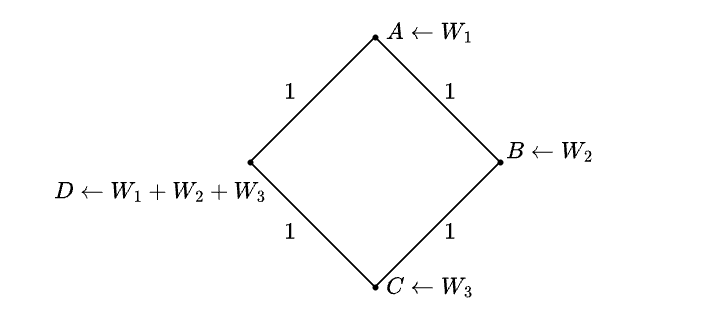}
    \caption{Example of a four-node network with three original files.}
    \label{fig:cycle-example}
\end{figure}

We now compute the average latency by determining the time required to recover each file at each node, and then averaging over all node--file pairs.

Node $B$ accesses $W_2$ locally with latency $0$, and retrieves $W_1$ and $W_3$ from its neighbors in time $1$.  
Node $A$ accesses $W_1$ locally, retrieves $W_2$ from $B$ in time $1$, and obtains $W_3$ in unit time by requesting $W_2$ from $B$ and $W_1+W_2+W_3$ from $D$ simultaneously, then canceling $W_1$ and $W_2$; node $C$ is symmetric.  
Node $D$ recovers $W_1$ and $W_3$ from its neighbors $A$ and $C$, and recovers $W_2$ by subtracting them from its locally stored coded file, all in unit time.  

Thus, each (node, file) pair has latency either $0$ or $1$, and a direct count shows that the resulting average latency is $3/4$.
\end{example}

\changed{We call an assignment \emph{uncoded} if every node stores a nonzero scalar multiple of one original file, and \emph{coded} otherwise. The assignment in Example~\ref{ex: sq} is therefore a coded assignment. }

A natural question is whether we can efficiently compute a latency-minimizing file assignment given the network graph. 
We show that the answer is negative for $k \geq 3$ files (Sec.~\ref{sec: impos}) via a reduction from the domatic number problem, which is NP-hard.

On the positive side, we show that certain graph classes (weighted trees, certain weighted cycles, and unit-weight graphs with large minimum degree) admit a type of uncoded, structured optimal assignment that we call a \emph{greedy assignment}: \changed{for every node, one can choose its $k$ closest storage locations, including the node itself, such that they collectively store all $k$ original files.}
We formalize this notion in Section~\ref{sec: model} and characterize the admitting classes in Section~\ref{sec: classes}.

% On the positive side, we \changed{identify several graph classes for which assignments that minimize latency can be constructed efficiently. These include weighted trees, certain weighted cycles, and graphs with unit edge weights and sufficiently large minimum degree. The resulting assignments are uncoded and have a simple local structure: for every node, one can choose its $k$ closest storage locations, including the node itself, so that together they store all $k$ original files. We formalize this notion in Section~\ref{sec: model} and characterize the admitting classes in Section~\ref{sec: classes}.}

\subsection{Related Work}

Our problem formulation builds on the framework introduced by Acharya, Kumar, and Cadambe in~\cite{acharya2024existence}, where the authors provide a graph-theoretic characterization of greedy assignments. 
Subsequent work~\cite{acharya2025chromatic} studies low average-latency codes among schemes achieving optimal per-node worst-case latency, while~\cite{acharya2025latency} incorporates file--node preference weights under weighted demand models. 
\changed{None of these works, however, address the computational complexity studied here or identify the structured graph families for $k\geq 3$ established in this paper.}

Related work on latency minimization in geo-distributed systems~\cite{pu2015low,malekimajd2015minimizing,narayanan2017right,uluyol2020near} focuses on system-level design aspects such as resource allocation and protocol optimization, rather than the fundamental limits and structure of file placement that we study here.

Another line of work considers graph-constrained storage codes~\cite{mazumdar2015storage,yu2014locally,patra2022node} with emphasis on worst-case metrics such as repair bandwidth and code distance. 
In contrast, we focus on average-case latency and structural properties governing latency-optimal assignments across network topologies.

We emphasize that redundancy in our setting serves to reduce retrieval latency, not to provide fault tolerance, distinguishing our work from distributed erasure coding~\cite{dimakis2010network,ramkumar2022codes}. 
Our problem shares similarities with coded caching~\cite{maddah2014fundamental,ji2015fundamental}, which also targets latency reduction under memory constraints. 
However, coded caching assumes a central server broadcasting coded content to users with local caches over a shared link, whereas we consider distributed storage across a weighted network graph with point-to-point file retrieval, optimizing average latency over arbitrary network paths.

\section{System Model and Latency Metrics}\label{sec: model}

We consider a network of $n$ servers, denoted by $V = \{1,2,\dots,n\}$, modeled as a connected weighted graph $G = (V, E)$. 
Each vertex represents a server and each edge weight specifies the (symmetric) round-trip communication latency between servers.

\changed{Nodes may relay data through intermediate servers. We therefore define $\tau(i,j)$ as the minimum, over all paths from $i$ to $j$, of the sum of the edge weights along the path.} Since $\tau$ defines a metric on the graph, we refer to $\tau(i,j)$ as the \emph{distance} between $i$ and $j$ for notational simplicity. We assume all edge weights are strictly positive, so $\tau(i,j) > 0$ for all $i \neq j$.

There are $k < n$ distinct information files $W_1, W_2, \dots, W_k$, each belonging to a vector space over a finite field. Each node $i \in V$ stores a single \emph{coded file} 
\[
    X_i = \sum_{j=1}^{k} g_{ij} W_j,
\]
where $g_{ij}$ are coefficients from the same finite field. An \emph{assignment} is a collection $X = (X_1, \dots, X_n)$. We call $X$ \emph{uncoded} if each $X_i$ is \changed{a nonzero scalar multiple of one original file}, and \emph{coded} otherwise. \changed{We assume that the stored files jointly span all $k$ original files.}

To retrieve file $W_j$, node $i$ contacts a subset $S \subseteq V$, downloads the coded files $\{X_\ell : \ell \in S\}$, and reconstructs $W_j$ via linear decoding. The retrieval latency is
\[
    \ell_j^{(i)} = \min_{S \subseteq V} \left\{ \max_{\ell \in S} \tau(i, \ell) \,:\, W_j \in \mathrm{span}\{X_\ell : \ell \in S\} \right\}.
\]
The \emph{average latency} of assignment $X$ is defined as
\begin{align}\label{eq:Lavg-def}
    L_{\mathrm{avg}}(X) 
    \;:=\; \frac{1}{nk}\sum_{i=1}^{n}\sum_{j=1}^{k} \ell_j^{(i)}.
\end{align}
% We also define the $p$-th power-mean latency for $p \ge 1$:
% \[
%     L^{(p)}(X)
%     \;=\; \left[\frac{1}{nk} 
%         \sum_{i=1}^{n} \sum_{j=1}^{k} \bigl(\ell_j^{(i)}\bigr)^p
%       \right]^{1/p}.
% \]

\paragraph{Distance profile and $k$-sites}
Fix a node $i \in V$. For $j \in \{1,\dots,n\}$, let $\lambda^{(i)}_j$ denote the distance from $i$ to its $j$-th closest node in $G$, with ties broken arbitrarily. By convention, $\lambda^{(i)}_1 = 0$ since node $i$ is closest to itself.

We introduce the following important definition.
\begin{definition}[$k$-site]\label{def:k-site}
\changed{A} \emph{$k$-site} of node $i$, denoted $S_k^G(i)$ (or simply $S_k(i)$), is \changed{any set of} $k$ closest nodes to $i$ in $G$, including $i$ itself.
\end{definition}

To reconstruct $j$ distinct files at node $i$, the node must download at least $j$ linearly independent coded files, which requires contacting at least $j$ distinct nodes. Thus, for all $i \in \{1,\dots,n\}$ and $j \in \{1,\dots,k\}$,
\begin{align}\label{eq:l-lambda-basic}
    \ell_{\pi^{(i)}_j}^{(i)} 
    \ge \lambda_j^{(i)}.
\end{align}
Here $\pi^{(i)}$ is a permutation of $\{1,\dots,k\}$ that orders the latencies at node $i$ non-decreasingly:
\[
    \ell_{\pi^{(i)}_1}^{(i)} \le \ell_{\pi^{(i)}_2}^{(i)} \le \cdots \le \ell_{\pi^{(i)}_k}^{(i)}.
\]
Summing over all nodes and files yields the lower bound from~\cite{acharya2024existence}:
\begin{align}\label{eq:Lavg-bound}
    L_{\mathrm{avg}}(X) 
    \;\ge\; 
    \frac{1}{nk}
    \sum_{i=1}^{n} \sum_{j=1}^{k} \lambda_j^{(i)},
\end{align}
which averages the distances within each node's $k$-site.

\paragraph{Greedy assignments}
In general, deciding whether a given assignment minimizes $L_{\mathrm{avg}}(X)$ is difficult. Therefore, we focus on assignments that achieve equality in~\eqref{eq:Lavg-bound}.

\begin{definition}\label{def:greedy}
An assignment $X$ is \emph{greedy} if for every node $i$ and every $j \in \{1,\dots,k\}$, $\ell_{\pi^{(i)}_j}^{(i)} = \lambda_j^{(i)}$, or equivalently, 
$\{\lambda_1^{(i)}, \dots ,\lambda_k^{(i)}\} = \{\ell_1^{(i)}, \dots, \ell_k^{(i)}\}$ as multisets.
\end{definition}

\begin{remark}
If a greedy assignment exists, it achieves the minimum possible average latency by~\eqref{eq:Lavg-bound}. By a similar argument, greedy assignments also minimize the $p$-th power-mean latency $L^{(p)}(X) := \left[\frac{1}{nk} \sum_{i=1}^{n} \sum_{j=1}^{k} (\ell_j^{(i)})^p\right]^{1/p}$ for all $p \ge 1$, as well as the per-node worst-case latency.
\end{remark}

The following characterization (adapted from Proposition~3 of~\cite{acharya2024existence}) explains the term ``greedy.''

\begin{proposition}\label{prop:greedy}
An assignment $X$ is greedy if and only if, for every node $i$, \changed{there exists a $k$-site $S_k(i)$ whose nodes store one uncoded copy of each original file}.
\end{proposition}

\begin{changedblock}
Since $\lambda_1^{(i)}=0$ for every node $i$, greediness implies that every node stores an uncoded file. For a fixed node $i$, choose a nearest node storing each original file. Under greediness, their distances from $i$ are exactly $\lambda_1^{(i)},\ldots,\lambda_k^{(i)}$ as a multiset, so these nodes form a $k$-site of $i$. Conversely, a $k$-site storing one uncoded copy of each original file attains equality in~\eqref{eq:l-lambda-basic} at $i$.
\end{changedblock}

%%%%%%%%%%%%%%%%%%%%%%%%%%%%%%%%%%%%%%%%%%%%%%%%%%%

\section{Hardness of Average-Latency Minimization}\label{sec: impos}

A natural question arising from our model is whether one can efficiently compute an assignment that minimizes average latency. Acharya et al.~\cite{acharya2024existence} showed that for $k = 2$ files, a greedy assignment always exists and can be computed in polynomial time.

\begin{changedblock}
We show that for every fixed $k \ge 3$, deciding whether a greedy $k$-file assignment exists is NP-complete. \footnote{In complexity theory, NP is the class of problems whose solutions can be verified in polynomial time. A problem is NP-hard if solving it efficiently would allow all NP problems to be solved efficiently. A problem is NP-complete if it is both in NP and NP-hard.} This immediately implies NP-hardness of exact average-latency minimization: an optimal assignment is greedy exactly when the lower bound~\eqref{eq:Lavg-bound} is attainable.
\end{changedblock}

\begin{theorem}\label{thm: impos}
\changed{For every fixed integer $k \ge 3$, determining whether a weighted network admits a greedy $k$-file assignment is NP-complete.}
\end{theorem}

We prove this by reduction from the \emph{domatic number} problem, a classical NP-hard problem~\cite{kaplan1994domatic}.

\begin{definition}[Dominating set and domatic number]
In an \emph{unweighted} graph $G=(V,E)$, a set $D \subseteq V$ is \emph{dominating} if every vertex is in $D$ or adjacent to a vertex in $D$. A \emph{domatic partition} is a partition $V = D_1 \cup \cdots \cup D_t$ of $V$ into pairwise disjoint dominating sets. The \emph{domatic number} $\dom(G)$ is the maximum $t$ for which such a partition exists.
\end{definition}

Equipped with these notions, we now move on to the proof of Theorem~\ref{thm: impos}.

\begin{proof}
To show NP-completeness, we establish membership in NP and NP-hardness.

Membership in NP is immediate: \changed{a certificate specifies a candidate
assignment and, for every node $u$, a $k$-site $S_k(u)$. We can verify
in polynomial time that each chosen $k$-site contains $k$ distinct
uncoded files. By Proposition~\ref{prop:greedy}, such a certificate
exists} if and only if the assignment is greedy.

For NP-hardness, fix $k \geq 3$ and consider an unweighted graph $G = (V,E)$. 
\changed{Let $\delta(G)$ denote the minimum degree of $G$. Since
$\dom(G)\leq \delta(G)+1$, every instance with
$\delta(G)<k-1$ is immediately a no-instance. Hence, without affecting
NP-hardness, we may restrict attention to graphs satisfying
$\delta(G)\geq k-1$.}

The graph $G$ admits a domatic partition of size $k$ if and only if the corresponding network graph obtained by assigning unit weights to each edge\changed{\footnote{If $G$ is disconnected, its components may be connected using sufficiently large-weight edges. Since $\delta(G)\geq k-1$, these edges do not affect any $k$-site.}} admits a greedy assignment of $k$ files. 

Indeed, in the unit-weight setting, \changed{since $\delta(G)\geq k-1$, every $k$-site consists of the node itself and $k-1$ of its neighbors. Thus,} a greedy assignment exists exactly when \changed{for every node $i$, there
exists a $k$-site $S_k(i)$ containing} $k$ distinct uncoded files, and grouping the vertices by the file they store yields a partition of $V$ into $k$ classes, each of which is a dominating set. Conversely, any domatic partition of size $k$ induces such an assignment by associating a distinct file with each partition class.

Therefore, deciding whether $\dom(G) \geq k$ is equivalent to deciding whether the \changed{corresponding weighted} version of $G$ admits a greedy assignment for $k$ files. In particular, any algorithm that determines the existence of a greedy $k$-file assignment can be used to decide whether $\dom(G) \geq k$. \changed{Since deciding whether $\dom(G)\geq k$ is NP-hard for every fixed $k\geq 3$~\cite{kaplan1994domatic}, this equivalence shows that deciding whether a graph admits a greedy assignment for $k$ files is NP-hard.}
\end{proof}

\begin{corollary}
\changed{For every fixed integer $k \geq 3$, computing an average-latency--minimizing assignment is NP-hard.}
\end{corollary}

\section{Networks that Admit Greedy Assignments}\label{sec: classes}

Although deciding the existence of greedy assignments in general network graphs is computationally intractable, this hardness does not extend to all network classes. In this section, we identify structured settings in which greedy assignments can be efficiently found. Specifically, we consider weighted tree networks, weighted cycle networks when $k$ divides $n$, and unit-weight graphs with sufficiently large minimum degree.

\subsection{Weighted tree networks}

We show that for tree networks with arbitrary edge weights, greedy assignments always exist and can be efficiently constructed.

\begin{theorem}\label{thm:tree}
    Let $G$ be a weighted tree network. Then, for any $k \le n$, the network supports a greedy file assignment, and such an assignment can be constructed in $O(n^2)$ time.
\end{theorem}

We begin with a general \changed{separator} lemma (not restricted to trees), which we will use to prove Theorem~\ref{thm:tree}.

\begin{changedblock}
\begin{lemma}\label{lem:separator}
Let $G=(V,E)$ be a weighted network graph, and suppose $V=V_1\cup V_2$, $V_1\cap V_2=\{w\}$, and every path between $V_1\setminus\{w\}$ and $V_2\setminus\{w\}$ passes through $w$.

If $u \in V_1$ and the $k$-site $S_k(w)$ is unique, then
\[
    S_k(u) \setminus S_k(w) \;\subseteq\; V_1 \setminus\{w\}.
\]
\end{lemma}
\end{changedblock}

\begin{proof}
Suppose toward a contradiction that there exists $u' \in S_k(u) \setminus S_k(w)$ with $u' \notin V_1 \setminus\{w\}$. 

We first observe that $u' \neq w$. Indeed, since $w$ is always its own closest neighbor, we have $w \in S_k(w)$. But $u' \notin S_k(w)$ by assumption, so $u' \neq w$. Therefore, since $u' \notin V_1 \setminus\{w\}$ and $u' \neq w$, we must have $u' \in V_2 \setminus \{w\}$.

Now let $w' \in S_k(w)$ be arbitrary. Since $u' \notin S_k(w)$ and $S_k(w)$ is unique, we have $\tau(w', w) < \tau(u', w)$.

Since $u' \in V_2 \setminus \{w\}$ and $u \in V_1$, any shortest path between $u$ and $u'$ must pass through \changed{the separator} $w$. This gives $\tau(u',u) = \tau(u',w) + \tau(w,u)$. By the triangle inequality, we also have $\tau(w',u) \le \tau(w',w) + \tau(w,u)$. Combining these two equations with the earlier strict inequality yields $\tau(w',u) < \tau(u',u)$.

This shows that every vertex $w' \in S_k(w)$ is strictly closer to $u$ than $u'$ is. Since $|S_k(w)| = k$, there are at least $k$ vertices strictly closer to $u$ than $u'$, contradicting the assumption that $u' \in S_k(u)$.
\end{proof}

Before proving Theorem~\ref{thm:tree}, we introduce some terminology. We model the given tree as a \emph{rooted tree} with a designated root vertex. For any non-root node $v$, the unique path to the root determines its \emph{parent} $P(v)$ (the first vertex on this path) and its \emph{predecessors} (all vertices on this path, including $P(v)$). Conversely, $v$ is a \emph{successor} of $u$ if $u$ is a predecessor of $v$.

We now proceed to the proof.

\begin{proof}
\begin{changedblock}
We first prove the result under the assumption that each node has a unique $k$-site, and then extend it to arbitrary edge weights.
\end{changedblock}

Let $v_1$ be the root of the tree, and label the remaining nodes $(v_i)_{i=2}^n$ so that, for every non-root vertex $v_i$, its parent $P(v_i)$ appears earlier in the sequence and every successor of $v_i$ appears later. For example, such an ordering can be obtained by exploring the tree from the root using a breadth-first or depth-first search and listing each vertex when it is first visited.

We construct a greedy assignment iteratively, processing the vertices in the order $v_1,v_2,\dots,v_n$. The outer loop runs over $i=1,\dots,n$ with $v_i$ as the current vertex, and for each $v_i$ the inner loop runs over the nodes in its $k$-site $S_k(v_i)$. For each $w\in S_k(v_i)$ we do:
\begin{itemize}
    \item if $w$ has not yet been assigned a file, assign to $w$ the smallest-index file not already assigned within $S_k(v_i)$;
    \item if $w$ has already been assigned a file, leave it unchanged.
\end{itemize}
We say that a vertex $v$ is \emph{processed} once the outer-loop iteration with $v$ as the current vertex has completed; in particular, after $v$ is processed, all elements of $S_k(v)$ have been assigned a file.

We first prove the following claim: after vertices $v_1,\dots,v_i$ have been processed,
\begin{itemize}
    \item[(a)] every element of $S_k(v_{i+1})\cap S_k(P(v_{i+1}))$ has been assigned a file, and
    \item[(b)] no element of $S_k(v_{i+1})\setminus S_k(P(v_{i+1}))$ has been assigned a file.
\end{itemize}

For (a), since $P(v_{i+1})$ appears earlier in the traversal order, it is processed before $v_{i+1}$. During the iteration for $P(v_{i+1})$, all vertices in $S_k(P(v_{i+1}))$ are assigned files. Therefore, all vertices in $S_k(v_{i+1})\cap S_k(P(v_{i+1}))$ are already assigned when we begin processing $v_{i+1}$.

For (b), since only $v_1,\ldots,v_i$ have been processed, it suffices to show that, for every $j\in\{1,\ldots,i\}$,
\begin{align}\label{eq:site-disjoint}
    S_k(v_j)\cap\bigl(S_k(v_{i+1})\setminus S_k(P(v_{i+1}))\bigr)=\emptyset.
\end{align}
Let $w=P(v_{i+1})$ and define
\[
    V_2:=\{v_{i+1}\}\cup\{\text{successors of }v_{i+1}\}\cup\{w\},
\]
and
\[
    V_1:=V\setminus(V_2\setminus\{w\}).
\]
(See Fig.~\ref{fig:tree-separator}.) \changed{By construction, $V=V_1\cup V_2$, $V_1\cap V_2=\{w\}$, and every path between $V_1\setminus\{w\}$ and $V_2\setminus\{w\}$ passes through $w$.} Moreover, $v_1,\ldots,v_i$ lie in $V_1$.

For any $j\in\{1,\ldots,i\}$, Lemma~\ref{lem:separator}, applied with $u=v_j\in V_1$, gives
\[
    S_k(v_j)\setminus S_k(w)\subseteq V_1\setminus\{w\}.
\]

\begin{figure}[H]
    \centering
    \includegraphics[width=1\linewidth, trim={5mm 2cm 5mm 1.6cm}, clip]{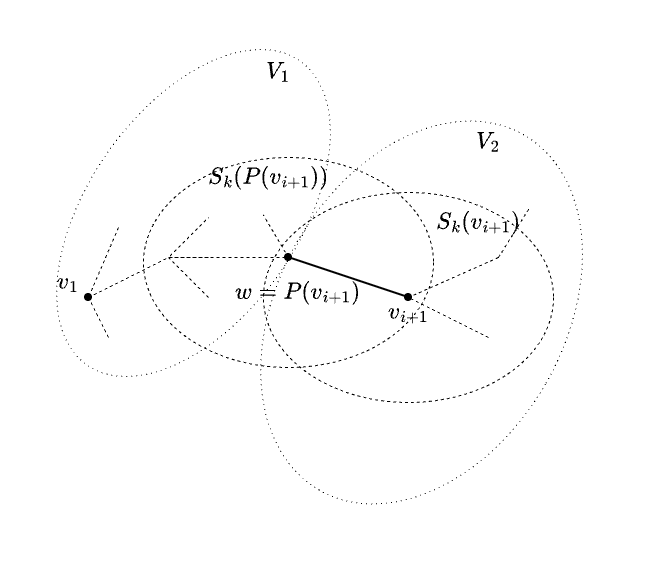}
    \caption{Illustration of $v_{i+1}$, $w$, \changed{the sets} $V_1,V_2$, and the $k$-sites.}
    \label{fig:tree-separator}
\end{figure}

Applying the lemma with the roles of $V_1$ and $V_2$ interchanged and $u=v_{i+1}\in V_2$ gives
\[
    S_k(v_{i+1})\setminus S_k(w)\subseteq V_2\setminus\{w\}.
\]
Since $V_1\setminus\{w\}$ and $V_2\setminus\{w\}$ are disjoint and $w=P(v_{i+1})$, this proves~\eqref{eq:site-disjoint}.

We now prove by induction that after processing $v_i$, the vertices in $S_k(v_i)$ store distinct files. For $i=1$, this holds because we assign all $k$ files within $S_k(v_1)$.

Assume it holds for $v_1,\ldots,v_i$, and consider $v_{i+1}$. The vertices in $S_k(v_{i+1})\cap S_k(P(v_{i+1}))$ already store distinct files by the induction hypothesis, while part (b) of the claim shows that the remaining vertices in $S_k(v_{i+1})$ are unassigned. Assigning them the unused files therefore makes all files in $S_k(v_{i+1})$ distinct, completing the induction.

\begin{changedblock}
It remains to consider a tree with nonunique $k$-sites. Since the tree has finitely many paths, we may perturb its edge weights arbitrarily slightly so that each node has a unique $k$-site while preserving every strict distance comparison under the original weights. Apply the preceding construction to the perturbed tree. For every node, its $k$-site under the perturbed weights remains a valid $k$-site when the original weights are restored, because the perturbation only breaks ties. Hence the same assignment is greedy for the original tree by Proposition~\ref{prop:greedy}.
\end{changedblock}

Finally, once the $k$-sites are precomputed, the assignment stage runs in $O(nk)$ time: the outer loop processes $n$ vertices once each, with each iteration scanning a $k$-site of size $k$. Including the preprocessing time to compute all $k$-sites, the total complexity is $O(n^2)$.
\end{proof}

%%%%%%%%%%%%%%%%%%%%%%%%%%%%%%%%
\subsection{Weighted cycle networks with $k$ dividing $n$}

We next consider weighted cycle networks. As shown in Example~\ref{ex: sq}, not all cycles admit greedy assignments. However, we obtain a simple sufficient condition.

\begin{proposition}
    Let $G$ be a cycle network on $n$ nodes. If $k \text{ divides } n$, then $G$ admits a greedy $k$-file assignment (regardless of edge weights), and such an assignment can be constructed in $O(n)$ time.
\end{proposition}

\begin{proof}
Assign to vertex $i$ the file indexed by \changed{$((i-1) \bmod k) + 1$}, so that the $k$ files repeat periodically along the cycle. For any node $i$, \changed{a $k$-site can be chosen as} $k$ consecutive nodes, which store distinct files by construction.
\end{proof}

\subsection{Unit-weight graphs with large minimum degree}

So far, the network classes we have considered were distinguished by their topology (trees, cycles). We now focus on unit-weight graphs and impose conditions on connectivity via the minimum degree.

As discussed in the proof of Theorem~\ref{thm: impos}, a unit-weight graph $G$ admits a greedy $k$-file assignment whenever $k \le \dom(G)$. Feige et al.~\cite[Theorem~1]{feige2000approximating} proved that for any graph $G$ with minimum degree $\delta$,
\[
    \dom(G) \;\ge\; (1-O(\ln \ln n /\ln n))\,\frac{\delta+1}{\ln n}.
\]
This yields the following.

\begin{proposition}
Let $G$ be a unit-weight network graph with minimum degree $\delta$. Then, for any
\[
    k \;\le\; (1-O(\ln \ln n /\ln n))\,\frac{\delta+1}{\ln n},
\]
the graph $G$ admits a greedy $k$-file assignment. 
\end{proposition}
Moreover, such an assignment can be found in polynomial time by computing a domatic partition in this regime~\cite{feige2000approximating}.

\section{Discussion and Future Work}

This work identifies several structured network classes (trees, certain cycles, and unit-weight graphs with large minimum degree) that admit greedy assignments.
A natural next step is to identify richer classes of network topologies that admit greedy assignments, particularly those reflecting practical geo-distributed deployments.

More broadly, our goal is to understand and approximate $L_{\mathrm{avg}}$-minimizing assignments in regimes where coding offers a meaningful advantage for latency reduction.
This includes deriving tighter bounds on optimal average latency and developing efficient algorithms for low-latency assignments in large-scale networks.

\vspace{2pt}
{\sc Acknowledgment:} This research was supported by NSF under award number 2516418.

\bibliographystyle{IEEEtran}
\bibliography{ref_geo_dist}

@article{kaplan1994domatic,
  author  = {Haim Kaplan and Ron Shamir},
  title   = {The Domatic Number Problem on Some Perfect Graph Families},
  journal = {Information Processing Letters},
  volume  = {49},
  number  = {1},
  pages   = {51--56},
  year    = {1994},
  doi     = {10.1016/0020-0190(94)90054-X}
}

@inproceedings{feige2000approximating,
  title={Approximating the domatic number},
  author={Feige, Uriel and Halld{\'o}rsson, Magn{\'u}s M and Kortsarz, Guy},
  booktitle={Proceedings of the thirty-second annual ACM symposium on Theory of computing},
  pages={134--143},
  year={2000}
}

@article{dimakis2010network,
  title={Network coding for distributed storage systems},
  author={Dimakis, Alexandros G and Godfrey, P Brighten and Wu, Yunnan and Wainwright, Martin J and Ramchandran, Kannan},
  journal={IEEE transactions on information theory},
  volume={56},
  number={9},
  pages={4539--4551},
  year={2010},
  publisher={IEEE}
}

@inproceedings{yu2014locally,
  title={Locally repairable codes over a network},
  author={Yu, Quan and Sung, Chi Wan and Chan, Terence H},
  booktitle={2014 IEEE Information Theory Workshop (ITW 2014)},
  pages={70--74},
  year={2014},
  organization={IEEE}
}

@article{maddah2014fundamental,
  title={Fundamental limits of caching},
  author={Maddah-Ali, Mohammad Ali and Niesen, Urs},
  journal={IEEE Transactions on information theory},
  volume={60},
  number={5},
  pages={2856--2867},
  year={2014},
  publisher={IEEE}
}

@article{mazumdar2015storage,
  title={Storage capacity of repairable networks},
  author={Mazumdar, Arya},
  journal={IEEE Transactions on Information Theory},
  volume={61},
  number={11},
  pages={5810--5821},
  year={2015},
  publisher={IEEE}
}

@article{pu2015low,
  title={Low latency geo-distributed data analytics},
  author={Pu, Qifan and Ananthanarayanan, Ganesh and Bodik, Peter and Kandula, Srikanth and Akella, Aditya and Bahl, Paramvir and Stoica, Ion},
  journal={ACM SIGCOMM Computer Communication Review},
  volume={45},
  number={4},
  pages={421--434},
  year={2015},
  publisher={ACM New York, NY, USA}
}

@article{malekimajd2015minimizing,
  title={Minimizing latency in geo-distributed clouds},
  author={Malekimajd, Marzieh and Movaghar, Ali and Hosseinimotlagh, Seyedmahyar},
  journal={The Journal of Supercomputing},
  volume={71},
  number={12},
  pages={4423--4445},
  year={2015},
  publisher={Springer}
}

@article{ji2015fundamental,
  title={Fundamental limits of caching in wireless D2D networks},
  author={Ji, Mingyue and Caire, Giuseppe and Molisch, Andreas F},
  journal={IEEE Transactions on Information Theory},
  volume={62},
  number={2},
  pages={849--869},
  year={2015},
  publisher={IEEE}
}

@inproceedings{narayanan2017right,
  title={Right-sizing geo-distributed data centers for availability and latency},
  author={Narayanan, Iyswarya and Kansal, Aman and Sivasubramaniam, Anand},
  booktitle={2017 IEEE 37th International Conference on Distributed Computing Systems (ICDCS)},
  pages={230--240},
  year={2017},
  organization={IEEE}
}

@inproceedings{uluyol2020near,
  title = {{Near-Optimal} Latency Versus Cost Tradeoffs in {Geo-Distributed} Storage},
  author={Uluyol, Muhammed and Huang, Anthony and Goel, Ayush and Chowdhury, Mosharaf and Madhyastha, Harsha V},
  booktitle={17th USENIX Symposium on Networked Systems Design and Implementation (NSDI 20)},
  pages={157--180},
  year={2020}
}

@article{ramkumar2022codes,
  title={Codes for distributed storage},
  author={Ramkumar, Vinayak and Balaji, SB and Sasidharan, Birenjith and Vajha, Myna and Krishnan, M Nikhil and Kumar, P Vijay},
  journal={Foundations and Trends in Communications and Information Theory},
  volume={19},
  number={4},
  pages={547--813},
  year={2022},
  publisher={Emerald Publishing Limited}
}

@article{patra2022node,
  title={Node repair on connected graphs},
  author={Patra, Adway and Barg, Alexander},
  journal={IEEE Transactions on Information Theory},
  volume={68},
  number={5},
  pages={3081--3095},
  year={2022},
  publisher={IEEE}
}

@article{emara2023geographically,
  title={Geographically distributed data management to support large-scale data analysis},
  author={Emara, Tamer Z and Trinh, Thanh and Huang, Joshua Zhexue},
  journal={Scientific Reports},
  volume={13},
  number={1},
  pages={17783},
  year={2023},
  publisher={Nature Publishing Group UK London}
}

@inproceedings{acharya2024existence,
  title={On existence of latency optimal uncoded storage schemes in geo-distributed data storage systems},
  author={Acharya, Srivathsa and Kumar, P Vijay and Cadambe, Viveck R},
  booktitle={2024 IEEE International Symposium on Information Theory (ISIT)},
  pages={1462--1467},
  year={2024},
  organization={IEEE}
}

@inproceedings{acharya2025chromatic,
  title={Chromatic Codes for Latency Optimal Geo-Distributed Storage},
  author={Acharya, Srivathsa and Kumar, P Vijay and Cadambe, Viveck R},
  booktitle={2025 IEEE International Symposium on Information Theory (ISIT)},
  pages={1--6},
  year={2025},
  organization={IEEE}
}

@inproceedings{acharya2025latency,
  title={Latency-optimal file assignment in geo-distributed storage with preferential demands},
  author={Acharya, Srivathsa and Kumar, P Vijay and Cadambe, Viveck R},
  booktitle={2025 IEEE Information Theory Workshop (ITW)},
  pages={644--649},
  year={2025},
  organization={IEEE}
}

@inproceedings{baron2026acos,
  title={{ACOS}: {Apple’s} Geo-Distributed Object Store at Exabyte Scale},
  author={Baron, Benjamin and Bousquet, Aline and Metens, Eric and Pimpale, Swapnil and Puz, Nick and de Saint Sauveur, Marc and Muzumdar, Varsha and Ari, Vinay},
  booktitle={24th USENIX Conference on File and Storage Technologies (FAST 26)},
  pages={53--66},
  year={2026}
}

\end{document}